\documentclass[preprint]{elsarticle} 

\usepackage{hyperref}
\usepackage{amssymb}
\usepackage{amsmath}
\usepackage{amsthm}
\usepackage{algorithm}
\usepackage{algpseudocode}

\newtheorem{theorem}{Theorem}
\newtheorem{lemma}{Lemma}

\newtheorem{assumption}{Assumption}

\begin{document}

\begin{frontmatter}

\title{Complexity of interval minmax regret scheduling on parallel identical machines with total completion time criterion}

\author[md]{Maciej Drwal}\ead{maciej.drwal@pwr.edu.pl}
\author[rr]{Roman Rischke}\ead{rischke@ma.tum.de}

\address[md]{Department of Computer Science, Wroc\l{}aw University of Technology,\\ Wybrze\.ze Wyspia\'nskiego 27, 50-370 Wroc\l{}aw, Poland}
\address[rr]{Center for Mathematics, Technische Universit{\"a}t M{\"u}nchen,\\ Boltzmannstr. 3, 85747 Garching bei M{\"u}nchen, Germany}

\begin{abstract}
	In this paper, we consider the problem of scheduling jobs on parallel identical machines, where the processing times of jobs are uncertain: only interval bounds of processing times are known. The optimality criterion of a schedule is the total completion time. In order to cope with the uncertainty, we consider the maximum regret objective and we seek a schedule that performs well under all possible instantiations of processing times. Although the deterministic version of the considered problem is solvable in polynomial time, the minmax regret version is known to be weakly NP-hard even for a single machine, and strongly NP-hard for parallel unrelated machines. In this paper, we show that the problem is strongly NP-hard also in the case of parallel identical machines.
	
\end{abstract}

	
\end{frontmatter}

\section{Introduction}

Robust optimization has been applied to many combinatorial optimization problems, since in practical applications input data to most problems can be rarely given precisely. This is true in the context of scheduling, as in many actual execution environments (e.g., computer systems, transportation, manufacturing) processing times of tasks are not known exactly, but their values can fluctuate within certain bounds. Moreover, very often no good assumptions can be made even regarding their probability distributions. In such circumstances we would like to find a solution that is the best in the worst possible scenario of events. Such solutions can be characterized in terms of the maximum regret criterion \cite{kouvelis1997robust, aissi2009min, kasperski2008discrete}. Solutions that minimize the maximum regret are often much more reliable than the ones obtained by ignoring parameter uncertainty. However, in many cases, finding a robust solution for uncertain data is more difficult and may require more computational resources.

We apply the minmax regret approach to the problem of scheduling on parallel identical machines to minimize the total completion time (sum of the completion times of all jobs) with interval uncertainty in the job processing times. The problem under consideration is denoted \textsc{interval minmax regret}~$P||\sum C_i$, using the notation from standard scheduling theory. The deterministic version of this problem can be solved in polynomial time by applying the {\it shortest processing time first} rule \cite{pinedo2012scheduling}. However, its minmax regret version becomes NP-hard even for a single machine, i.e., \textsc{interval minmax regret}~$1||\sum C_i$. In \cite{lebedev2006complexity} it is shown that even when the midpoints of all intervals are equal and the number of jobs is odd, finding an optimal robust sequence on a single machine is weakly NP-hard. Surprisingly, for an even number of jobs this problem is polynomially solvable. Thus the case in which the number of machines is given as part of the input can be no easier. Recently, Conde \cite{conde2014mip} indicated a simple reduction from the minmax regret assignment problem \cite{aissi2005complexity} of $m$ jobs to $m$ machines, which implies that in case of $m$ parallel unrelated machines (\textsc{interval minmax regret}~$R || \sum C_i$) the problem is strongly NP-hard. However, this reduction does not suffice to prove hardness for the case of parallel identical machines. In this paper, we extend the aforementioned complexity results by showing that the problem is strongly NP-hard even on identical parallel machines.

\section{Problem formulation}

In the scheduling problem $P || \sum C_i$ we are given $m$ identical parallel machines (processors) for processing $n$ jobs, where each job $i$ has an integer processing time $p_i$, $i=1,\ldots,n$.
If not stated otherwise, we assume that all $p_i \geq 0$. Each job has to be assigned to exactly one machine. Let $\pi_j$ denote a vector of integers, where $\pi_j(k)$ is the index of the job scheduled on the $j$th machine as the $k$th to the last (jobs on each machine are scheduled starting from time zero and without idle times). Let $n_j$ denote the number of elements in $\pi_j$, i.e., the number of jobs assigned to machine $j$. 
The {\it completion time} of the job scheduled as $k$th to the last on machine $j$ is $C_{j,k} = \sum_{i=k}^{n_j} p_{\pi_j(i)}$ ($C_{j,k} = 0$ if there is no such job). 
The objective is to minimize the total sum of completion times (also called the {\it total flow time}), expressed as:
\begin{equation}\label{psumci-obj}
	F(\pi) = \sum_{j=1}^m \sum_{k=1}^{n_j} C_{j,k} = \sum_{j=1}^m \sum_{k=1}^{n_j} k p_{\pi_j(k)},
\end{equation}
where $\pi = [\pi_1, \ldots, \pi_m]$ is called a {\it schedule}. We will sometimes refer to this problem formulation as the {\it deterministic version} of the scheduling problem.

The definition of the minmax regret version of this problem with interval uncertainty, denoted as \textsc{interval minmax regret} $P || \sum C_i$, differs as follows. Instead of having exact processing times $p_i$, we are now given only intervals~$[p^-_i, p^+_i]$, $i=1,\ldots,n$, to which the actual processing times belong. We denote by $S = [p_1^S, \ldots, p_n^S]$ any vector that satisfies $p^-_i \leq p_i^S \leq p^+_i$ for all $i=1,\ldots,n$. Such a vector will be called a {\it scenario}. For any schedule $\pi$ and scenario $S$ we define the value of {\it regret} as $Z(\pi, S) := F(\pi, S) - F^*(S)$, where $F(\pi, S)$ is the objective function \eqref{psumci-obj} from the deterministic version of the problem $P || \sum C_i$ with input data $S$, and $F^*(S)$ is the value of an optimal solution of this problem. The objective of \textsc{interval minmax regret} $P || \sum C_i$ is to minimize over schedules the maximum of regret over scenarios:
\begin{equation}\label{rob-obj}
	Z^* = \min_{\pi} \max_S \left( F(\pi, S) - F^*(S) \right).
\end{equation}
The above minmax regret formulation is a {\it robust optimization} formulation of the scheduling problem. A schedule minimizing the maximum regret will be called {\it robust optimal}.

\section{Computation of maximum regret}

The deterministic version of the considered problem is solvable in polynomial time (see \cite{pinedo2012scheduling}, Theorem 5.3.1). 
An optimal schedule can be obtained by first sorting all $n$ jobs in order of non-decreasing processing times, and then we assign the first unassigned job in the list to the least loaded machine, i.e., to the machine with the smallest current makespan. 
Repeating this procedure until all jobs are assigned gives the desired schedule.

We show that for a fixed schedule $\pi$ it is possible to compute the value of maximum regret $Z(\pi) = \max_S Z(\pi, S)$ in polynomial time. The method is analogous to the one presented in \cite{conde2014mip} for parallel unrelated machines, with the main difference that in the identical machines case the input data contains a single interval $[p_i^-, p_i^+]$ instead of $m$ intervals given in the case of unrelated machines. Thus we omit the simple proof of correctness of Formulas \eqref{max-regret}--\eqref{wc-scenario}.

Let us encode a feasible solution $\pi$ of the considered problem in terms of binary variables ${\bf x}$ as follows: let $x_{ijk} = 1$ iff the $i$th job is processed on the $j$th machine as the $k$th to the last.

For any feasible schedule ${\bf x}$ the maximum regret can be computed as:
\begin{equation}\label{max-regret}
	\sum_{j=1}^m \sum_{i=1}^n \sum_{k=1}^n k p_i^+ x_{ijk} - \min_{\bf y} \sum_{j=1}^m \sum_{i=1}^n \sum_{k=1}^n c_{ijk}({\bf x}) y_{ijk},
\end{equation}
where
\begin{equation}\label{cijk}
	c_{ijk}({\bf x}) = k p_i^- + (p_i^+ - p_i^-) \sum_{l=1}^m \sum_{r=1}^n \min\{ r, k \} x_{ilr}.
\end{equation}
The minimization in \eqref{max-regret} is equivalent to the minimum assignment problem and thus can be solved in polynomial time, using e.g. the Hungarian method \cite{papadimitriou1998combinatorial}. Variable ${\bf y}$ is an $n \times (mn)$ permutation matrix.

Let $x_{ij} = k$ iff $x_{ijk} = 1$, and $y_{ij} = k$ iff $y_{ijk} = 1$. For a fixed ${\bf x}$, given a solution ${\bf y}$ of the minimization in \eqref{max-regret}, the worst-case scenario can be obtained as:
\begin{equation}\label{wc-scenario}
	p_i = \left\{ \begin{array}{ll} p_i^+ & \text{if } \sum_{j=1}^m (x_{ij} - y_{ij}) \geq 0, \\ p_i^- & \text{otherwise.} \end{array} \right. 
\end{equation}

\section{Properties of optimal solutions}

Denote by $Z({\pi})$ the solution value of the \textsc{interval minmax regret} $P||\sum C_i$ problem for a schedule $\pi = [ \pi_1, \ldots, \pi_m]$. 
The maximum regret can be written as:
$$
	Z(\pi) = \max_S \left( F(\pi, S) - F^*(S) \right) = \max_S \sum_{j=1}^m \left( F_j(\pi_j, S) - F_j^*(S) \right) = \sum_{j=1}^m Z_j(\pi_j).
$$
Here $F_j(\pi_j, S) = \sum_{k=1}^{n_j}k p_{\pi_j(k)}^S$ is the sum of completion times of the jobs on machine~$j$ under scenario $S$, $n_j = |\pi_j|$ is the number of jobs assigned to machine~$j$, $F_j^*(S)$ is the sum of completion times of the jobs on machine~$j$ in an optimal solution under scenario $S$, and finally, $Z_j(\pi_j)$ is the difference between $F_j(\pi_j, S)$ and $F_j^*(S)$ for a scenario $S$ that maximizes the total regret. Let $\pi^*$ denote an optimal robust solution, i.e., a schedule that minimizes $Z$. 

Consider a worst-case scenario for any schedule $\pi$ (see Eq.~\eqref{wc-scenario}). This scenario defines an instance of the deterministic version of the problem. An optimal solution of this deterministic problem is called a {\it worst-case alternative} for $\pi$.

From now on, we consider only instances satisfying the following assumption.

\begin{assumption}\label{a1}
	The number of jobs $n$ is divisible by the number of machines~$m$, i.e., there exists an integer $n_0 > 0$ such that $n = m \cdot n_0$.
\end{assumption}

In particular, if $m | n$, then any schedule has a worst-case alternative with an equal number of jobs on each machine. Moreover, the following is true.

\begin{lemma}\label{lem:1}

If $m | n$, then in an optimal robust schedule every machine is assigned the same number of jobs.

\end{lemma}

\begin{proof}

Let $\pi_2$ be any schedule with different number of jobs on at least two machines. Under a fixed scenario $S$, there exists a schedule $\pi_1$ with an equal number of jobs on each machine, such that $F(\pi_1, S) < F(\pi_2, S)$. This follows from the fact that we can construct $\pi_1$ from $\pi_2$ by performing a sequence of the following job displacements: from the machine with the longest schedule remove the job from the first position and insert it at the first position on the machine with the least number of jobs (the multipliers $k$ in \eqref{psumci-obj} of the remaining jobs are unchanged, but the multiplier of the moved job may decrease; the last such displacement must be performed between two machines that differ in the number of jobs by 2, thus the multiplier of that job decreases by 1, and the overall cost of the schedule decreases). 

Let us denote by $S^{\pi}$ the worst-case scenario for $\pi$. Then we get:
\begin{align}
	Z(\pi_1) = F(\pi_1, S^{\pi_1}) - F^*(S^{\pi_1}) < F(\pi_2, S^{\pi_1}) - F^*(S^{\pi_1}) \nonumber
\\
	\leq F(\pi_2, S^{\pi_2}) - F^*(S^{\pi_2}) = Z(\pi_2). \nonumber
\end{align}

The last inequality follows from the fact that $S^{\pi_1}$ is not necessarily the worst-case scenario for $\pi_2$, thus by definition the value of regret $Z(\pi_2, S^{\pi_1})$ is no greater than that of the maximum regret, computed for the worst-case scenario $S^{\pi_2}$. Lemma~\ref{lem:1} follows from the above reasoning, which indicates that for any schedule with a different number of jobs on machines, there exists a strictly better solution with an equal number of jobs on each machine, so only such a solution could be optimal.

\end{proof}

Justified by Lemma~\ref{lem:1}, we restrict in the rest of the paper to schedules with an equal number of jobs on all machines, i.e., we consider $m \times n_0$ matrices, with each row representing a sequence of jobs on a machine. In order to prove the main result, we need the following lemmas.

\begin{lemma}\label{lem:exchanges}
	Given a schedule $\pi$, where $\pi$ is an $m \times n_0$ matrix, consider a schedule $\pi'$ obtained by switching a pair of elements in any column of $\pi$. Both schedules have the same maximum regret.
\end{lemma}

\begin{proof}
	Let $k^{\pi(i)}$ be the position of the $i$th job on its machine, counting from the last position on the machine (i.e., if $i$ happens to be scheduled on machine $j$, then the last position is $n_j$, and $1 \leq k^{\pi(i)} \leq n_j$). Since every job has to be assigned to exactly one position on one machine, then clearly the cost of the schedule $\pi$ under the scenario $S$ is $F(\pi, S) = \sum_{i=1}^n k^{\pi(i)} p_i^S$. Let $\sigma$ be the worst-case alternative for $\pi$, and let $k^{\sigma(i)}$ be the position to the last of the $i$th job on its machine in the worst-case alternative. Let $S^{\pi}$ be the worst-case scenario. Then the maximum regret can be expressed as:
\begin{equation}\label{mr-nomach}
		Z(\pi) = F(\pi, S^{\pi}) - F(\sigma, S^{\pi}) = \sum_{i=1}^n \left( k^{\pi(i)} - k^{\sigma(i)} \right) p_i^{S^{\pi}}. 
\end{equation}
	The worst-case scenario $S^{\pi}$ can be found by taking $p_i^{S^{\pi}} = p_i^+$ if $k^{\pi(i)} - k^{\sigma(i)} \geq 0$, and $p_i^{S^{\pi}} = p_i^-$ otherwise. Values $k^{\pi(i)}$ and $k^{\sigma(i)}$ can be easily determined given $\pi$ and $\sigma$.
	
	Observe that in the Eq.~\eqref{mr-nomach} there are no machine indices associated with jobs. The value of the maximum regret can be computed knowing only positions of jobs on machines, while the assignment of jobs to machines is irrelevant. Thus we may arbitrarily permute jobs within columns of $\pi$ obtaining schedules with the same maximum regret and the same worst-case alternative.
	
\end{proof}

\begin{lemma}\label{lem:optimal-form}

There exists an optimal robust schedule of \textsc{interval minmax regret} $P||\sum C_i$ such that for any machine $j=1,\ldots,m$ the schedule on machine~$j$ is the same as the optimal robust schedule in the \textsc{interval minmax regret} $1||\sum C_i$ problem.

\end{lemma}

\begin{proof}
According to Lemma~\ref{lem:exchanges}, given any schedule $\pi$, it is possible to construct an equivalent schedule $\pi'$ by permuting jobs within a column of $\pi$ arbitrarily. 
The worst-case scenario for $\pi'$ is identical to the worst-case scenario of $\pi$ due to Lemma~\ref{lem:exchanges} and Eq.~\eqref{wc-scenario}. 
The worst-case alternative for $\pi'$ can be obtained by sorting all $m \cdot n_0$ processing times of the worst-case scenario for $\pi'$ in an ascending order, then grouping them into $n_0$ consecutive sets of $m$ numbers, so that jobs from the same group occupy the $i$th position on a machine, $i= 1, \ldots, n_0$ (again, jobs within a group may be assigned to machines arbitrarily without changing the value of the solution). 
By $\sigma$ we denote the worst-case alternative for $\pi$.

Given a schedule $\pi$ and its worst-case alternative $\sigma$, we construct a schedule $\pi'$ and its worst-case alternative $\sigma'$, such that:
\begin{enumerate}[(i)]
	\item solution $\pi'$ has the same maximum regret as $\pi$ and
	\item for each machine $j = 1,\ldots,m$, each job processed in $\pi'$ on the $j$th machine is also processed on the $j$th machine in its worst-case alternative $\sigma'$.
\end{enumerate}
See Fig.~\ref{fig:1} for an example. If we can prove that such a construction always exists, then we can argue that we may restrict to the single machine problem.

To accomplish this, we construct an $n_0\times n_0$ matrix $M$ with the following properties. 
Let $x_i$ and $y_i$ be the position (column) of job $i$ in $\pi$ and $\sigma$, respectively.
Then job $i$ is inserted into $M$ at $(x_i,y_i)$.
Note that several jobs can be at the same position in $M$, but in every row and in every column of $M$ we have exactly $m$ jobs.
Empty positions in $M$ are marked with an empty sign.
Matrix~$M$ tells us that if two jobs $i_1$ and $i_2$ are in the same row or column of $M$, then they are in conflict, i.e., $i_1$ and $i_2$ cannot run on the same machine in $\pi'$ and $\sigma'$.
The question is now if we can color the jobs in $M$ with $m$ colors such that the jobs of each row and of each column have different colors. 
If two jobs have the same color, then they run on the same machine in $\pi'$ and $\sigma'$, the position remains the same as in $\pi$ and $\sigma$.

We give a positive answer and show it by induction on $m$.
Clearly, if we have only one machine, then we need only one color (machine).
Suppose the statement holds for $m-1$ machines. 
Assume for the moment that we can choose a set of jobs from $M$ so that we have exactly one job in each row and each column.
Removing the chosen jobs from $M$ gives the $m-1$ case, which by induction can be $(m-1)$-colored. 
The chosen jobs can be colored with a new color, yielding an $m$-coloring.

Now we need to show that it is possible to select a set of jobs so that we have exactly one job in each row and each column of $M$. 
This can be shown by Hall's marriage theorem~\cite{hall1935}. 
We associate with our matrix $M$ the following bipartite graph. 
We create a vertex for each row and each column of $M$ and connect the vertex for row $i$ with the vertex for column $j$ if the matrix $M$ has at least one element at position $(i,j)$. 
If this graph has a perfect matching, then we can select a set of jobs so that we have exactly one job in each row and each column. 
Assume, by contradiction, that this graph has no perfect matching. 
Then, by Hall's marriage theorem, there is a set $R$ of row vertices such that the neighbouring set $N(R)$ of column vertices is strictly smaller, i.e., $|N(R)| < |R|$. 
Consider the total number $\tau$ of jobs being in the rows and columns which are represented by $R$ and $N(R)$, i.e., we consider the submatrix formed by rows $R$ and columns $N(R)$. 
In every row from $R$ there are $m$ jobs and, by the definition of the graph, they are all in the columns associated with $N(R)$. 
That is, $\tau = m |R|$. 
In every column of $N(R)$ there are at most $m$ jobs when restricting to rows in $R$. 
Therefore, we have that $\tau \leq m |N(R)|$, which contradicts our assumption that $|N(R)| < |R|$, since $m>1$.  

Let $\pi^*$ be an optimal robust solution. 
We have shown that it is possible to obtain a schedule $\pi'$ which has the same total processing time as $\pi^*$ and the property that all jobs from the $j$th machine in $\pi'$ also appear on the $j$th machine in the worst-case alternative for $\pi'$. 
Thus $\pi'$ is also optimal robust and we may consider independently each schedule $\pi'_j$ on each $j$th machine. 
For any such machine the minimum of $Z_j(\pi_j) \geq 0$ is obtained if $\pi_j$ is the optimal robust solution of the single machine scheduling problem \textsc{interval minmax regret} $1|| \sum C_i$.
\end{proof}

\begin{figure}[!ht]
	\caption{Example of a pair of schedule $\pi$ and its worst-case alternative $\sigma$, along with a corresponding pair $\pi'$, $\sigma'$ satisfying the condition of the Lemma~\ref{lem:optimal-form}; $m=4$ machines, $n_0=4$ jobs on each machine (numbers denote the indices of jobs).}\label{fig:1}
	\centering
$$
	\begin{array}{cc}
		\pi = \left[
		\begin{array}{cccc}
			1 & 5 & 9 & 13\\
			2 & 6 & 10 & 14\\
			3 & 7 & 11 & 15\\
			4 & 8 & 12 & 16\\
		\end{array}
		\right]
		&
		\sigma = \left[
		\begin{array}{cccc}
			1 & 2 & 3 & 4\\
			5 & 6 & 7 & 8\\
			9 & 10 & 11 & 12\\
			13 & 14 & 15 & 16\\
		\end{array} 
		\right]

		\\
		\\
		
		\pi' = \left[
		\begin{array}{cccc}
			1 & 6 & 11 & 16\\
			2 & 5 & 12 & 15\\
			3 & 8 & 9 & 14\\
			4 & 7 & 10 & 13\\
		\end{array}
		\right]
		&
		\sigma' = \left[
		\begin{array}{cccc}
			1 & 6 & 11 & 16\\
			5 & 2 & 15 & 12\\
			9 & 14 & 3 & 8\\
			13 & 10 & 7 & 4\\
		\end{array} 
		\right]
		
	\end{array}
$$
\end{figure}

Let us take a closer look at the \textsc{interval minmax regret} $1||\sum C_i$ problem. Observe that its formulation remains valid when the bounds of intervals $p_i^-$ and $p_i^+$ are arbitrary (possibly negative) integers, and that adding the same constant to all bounds of intervals does not change the value of the maximum regret \cite{lebedev2006complexity}.

Consider instances with equal midpoints $c$ of all job processing time intervals, that is $[p_i^-, p_i^+] = [c - p_i, c + p_i]$, for all $i=1,\ldots,n$. Due to \cite{lebedev2006complexity} we know that an optimal solution for a single machine can be obtained as:

(a) for an even number of jobs $n=2k$ on a machine, $Z(\pi^*) = k\sum_i p_i$, 

(b) for an odd number of jobs $n=2k+1$ on a machine we have:
\begin{equation}\label{r2}
	Z(\pi^*) = k\sum_{i=1}^{n} p_i + \max\{ P_1, P_2 \},
\end{equation}    
where $(P_1, P_2)$ is a solution of the optimization version of the \textsc{balanced-partition} problem, i.e., $P_1$ and $P_2$ are the sums of two disjoint $k$-element subsets of the set of $2k$ smallest values $p_i$, and the value $|P_1 - P_2|$ is minimal among all such 2-partitions. The job with the widest interval is always inserted in the middle of the permutation and does not appear in $P_1$ or $P_2$. The remaining jobs are scheduled in such a way that the wider the interval, the closer the job is to the middle of the permutation (in \cite{lebedev2006complexity} authors call such schedules {\it uniform}).

\section{Problem complexity}

The main result in this paper, stated as Theorem \ref{thm:main}, is based on the reduction from a variant of a set partitioning problem that is strongly NP-complete. 

Consider an instance of the 3-\textsc{partition} problem: given is a set of $3m$ positive integers $a_i$, $i=1,\ldots,3m$, and an integer $B$, such that $\sum_{i=1}^{3m} a_i = mB$, and for all $i$, $B/4 < a_i < B/2$. The question is: can we partition the given set of integers into $m$ disjoint triplets of integers, such that each triplet sums up to exactly $B$?

We define the 4-\textsc{partition-into-pairs} problem (the 4-\textsc{pp} problem for short) as follows: given is a set of $4m$ positive integers $a_i$, $i=1, \ldots, 4m$. The question is whether it is possible to partition the given set of integers into $m$ disjoint quadruplets of integers $A_1, \ldots, A_m$, such that there exists a bijective function $f : \{ 1, \ldots, m \} \rightarrow \{ 1, \ldots, m\}$, such that:
$$
	\forall_{i \in \{ 1, \ldots, m \} } \; s(A_i) = s(A_{f(i)}) \; \textrm{ and } \; f(i) \neq i \; \textrm{ and } \; f(f(i)) = i,
$$
where $s(A)$ is the sum of elements in $A$. In other words, we want to partition the set of integers into $m$ 4-sets in such a way that all the 4-sets can be arranged in distinct pairs of equal sums.

\begin{theorem}\label{thm:4pp}
	The 4-\textsc{pp} problem is strongly NP-complete.
\end{theorem}

\begin{proof}

We reduce 3-\textsc{partition} to 4-\textsc{pp}. Given an instance of 3-\textsc{partition} with $3m$ elements and the target sum $B$, let us consider the instance of 4-\textsc{pp} with the following input data:

(a) all integers from the instance of 3-\textsc{partition}, denoted $a_1, a_2, \ldots, a_{3m}$,

(b) for $j = 1, 2, \ldots, m$, an integer  $5^{j+1}B$,

(c) for $j = 1, 2, \ldots, m$, a group of four integers $(5^{j+1} + 1)B/4$.

In the obtained instance we have $8m$ integers and the solution consists of $2m$ disjoint quadruplets.

Suppose the instance of 3-\textsc{partition} is positive. Denote its solution ($m$ triplets of integers): $A_1, A_2, \ldots, A_m$. Then clearly the obtained instance of 4-\textsc{pp} is positive, since we take each triplet of integers $A_j$ and add one integer $5^{j+1}B$ from the set (b) to obtain a 4-set with the sum $s(A_j \cup \{ 5^{j+1}B \}) = (5^{j+1} + 1)B$, for $j=1,\ldots,m$. The corresponding 4-set with an equal sum is just the four integers $(5^{j+1} + 1)B/4$ from the set (c).

Suppose the instance of 3-\textsc{partition} is negative. Then in the set (a) in every possible combination of triplets there would always be at least one triplet with a sum different than $B$, strictly between $3B/4$ and $3B/2$. We argue that it is not possible to match every 4-set with another one, so that these matches are all distinct, and both 4-sets in every pair have equal sums.

Given a 4-set $S$, if there exists a 4-set $S'$ such that $s(S) = s(S')$, we say that $S'$ is a {\it match} for $S$.

We first show that no match is possible for any 4-set containing more than one (b)-element. To see this, consider the representations of the given numbers in numeral system with the base 5, after dividing all the numbers by $B$. Each element from (b) is represented by a distinct digit 1. Numbers $\frac{1}{4}5^k$ are represented by the first $k$ digits 1 in the integral part, numbers $\frac{1}{2}5^k$ are represented by the first $k$ digits 2 in the integral part, and numbers $\frac{3}{4}5^k$ are represented by the first $k$ digits 3 in the integral part, where $k \geq 2$ (all these numbers also have digits in their fractional parts that continue indefinitely). For example, $(5^k + 1)/4$ is represented by $(11 \ldots 1.\bar{2})_5$, where there are $k$ digits 1 in the integral part.

Sum of any number of (b)-elements with up to 3 (c)-elements cannot result in the carry in integral digits, except at the first integral digit and in the fractional digits. Since all $a_i < B/2$, adding elements from the set (a) may affect only the first integral digit and the fractional digits.

It is clear that no match is possible for any 4-set consisting of only (b)-elements, since they are all distinct powers of 5.                                                                                                              

Moreover, if there are three (b)-elements and one element from either (a) or (c) in $S$, then a match would require that $S'$ contains the same four elements as $S$, but the (b)-elements are all unique. 

Denote the elements from (b) by $b(j) = 5^{j+1}B$ and the elements from (c) by $c(j) = (5^{j+1}+1)B/4$. If there are two (b)-elements in $S$, then in order to find a match we need to have $S'$ without any (b)-elements. The only way to accomplish this is to use a (c)-element in $S'$ that can be expressed as the sum of a (b)-element and a (c)-element in $S$, that is, $c(j+1) = c(j) + b(j)$. But then we have $S = \{ b(j), b(j+1), c(j), a_{i_1} \}$ and $S' = \{ c(j+2), a_{i_2}, a_{i_3}, a_{i_4} \}$, and since $B/4 < a_i < B/2$, then no such match is possible. The same is true if we replace $a_{i_1}$ in $S$ by a (c)-element, as $S'$ would then have to contain another (c)-element.

Consequently, only single (b)-elements in any 4-set $S$ are possible. But again, examining a base-5 representation of $s(S)$ containing element $b(j)$, we can see that a match for $S$ must have the same digit at the $(j+1)$-th position set. This can be obtained only in two ways: 

(i) by using $c(j+1)=b(j)+c(j)$ and matching $S$ with a set $S'$ containing element $c(j+1)$,

(ii) by matching $S$ with $S'$ consisting of four elements $c(j)$. 

Suppose we have one match of the former type for a 4-set $S$ containing $b(j)$. Then element $c(j+1)$ cannot be used in a match of the second type, because all four elements $c(j+1)$ are needed for that and one has already been used. But then $b(j+1)$ needs to be matched to a set containing $c(j+2)$. Continuing this reasoning, we conclude that we would need to use the first type of match for a set with the element $b(m)$, but this is impossible, since element $c(m+1)$ does not exist. 

Thus a set $S$ containing $b(j)$ needs to be matched to $S'$ consisting of four elements $c(j)$, for all $j = 1,\ldots,m$. Then $s(S')$ has exactly two digits 1 in the base-5 representation: the first digit, and the $(j+1)$-th digit. But to have the first digit 1 in $s(S)$, there must be three (a)-elements in $S$ that sum up to $B$.

Since there are $m$ (b)-elements, then $m$ 4-sets must contain a single (b)-element. But that requires that there are also $m$ triplets of (a)-elements that sum up to $B$ each. This contradicts the assumption that there is no solution for the considered instance of the 3-\textsc{partition} problem.

\end{proof}

\begin{theorem}\label{thm:main}
	The \textsc{interval minmax regret} $P||\sum C_i$ problem, in which the number of machines is given as a part of the input, is strongly NP-hard.
\end{theorem}

\begin{proof}

Given an instance of 4-\textsc{pp} with a set of $4m'$ integers $a_1, \ldots, a_{4m'}$, let us construct an instance of the \textsc{interval minmax regret} $P||\sum C_i$ problem with $m=m'/2$ machines and $n = 4m' + m'/2$ jobs (we assume w.l.o.g. that $m'$ is even): for each integer $a_i$ in the input data of 4-\textsc{pp} we add a job with a processing interval $[B-a_i, B+a_i]$ to the set of jobs, and additionally, we add $m$ jobs with processing intervals $[0, 2B]$, where $B$ is an integer greater than the sum of 4 largest integers in the given input data.

Each processing interval is a subset of $[0,2B]$, and midpoints of all processing intervals are the same and equal to $B$., i.e., $(p_i^- + p_i^+)/2 = B$.

We use Lemma~\ref{lem:optimal-form} and property \eqref{r2}, which imply that a single machine with 9 jobs yields the optimal maximum regret:
$$
	Z_j(\pi_j) = 4(\sum_{i=1}^9 a_{\pi_j(i)}) + \max \{ s(A_{j1}), s(A_{j2}) \},
$$
where $s(A_{j1}) = a_{\pi_j(1)} + a_{\pi_j(2)} + a_{\pi_j(3)} + a_{\pi_j(4)}$, and $s(A_{j2}) = a_{\pi_j(6)} + a_{\pi_j(7)} + a_{\pi_j(8)} + a_{\pi_j(9)} $.

Denote $C = \sum_{i=1}^{4m'}a_i$. We claim that the instance of 4-\textsc{pp} has a solution if and only if $Z(\pi^*) = 4mB + 9C/2$.

Assume that there exists a solution of 4-\textsc{pp}.

We know that for each single machine, given 9 jobs on that machine, the minimal maximum regret is obtained when the widest job is in the middle of the schedule (position 5), and 4-element subsets on both sides of the widest job have equal sums (this is always possible if there exists a solution of 4-\textsc{pp}: we take any two quadruplets of jobs with matching sums of values $a_i = p_i^+ - B = p_i^- + B$). 

We show that in an optimal robust solution exactly one job $[0, 2B]$ must be assigned to each machine (such job is consequently always in the middle of each permutation). Then given two 4-job sets with an equal sum $B_j = s(A_{j1}) = s(A_{j2})$ on each machine $j$, with $B_j \leq B$, we obtain the total maximum regret:
$$
	Z^* = \sum_{j=1}^m \left( 4(2B_j + B) + B_j) \right) = 4mB + 9 \sum_{j=1}^m B_j = 4mB + 9C/2.
$$

Suppose for the sake of contradiction that a schedule in which at least one machine is not assigned any $[0,2B]$ jobs is optimal robust. By the pigeonhole principle, there must be a machine with more than one $[0,2B]$ job. Denote by 1 the machine with more than one $[0,2B]$ jobs, and by 2 the machine with no $[0,2B]$ jobs. 

Let us construct another schedule by exchanging one $[0,2B]$ job from machine 1 with job $[B-a, B+a]$ from machine 2, $a < B$.

Let $Z_1$ be the regret on machine 1 before the exchange, and $Z_1'$ be the regret on machine 1 after the exchange:
\begin{align}
	Z_1 = 4(s(A_{11}) + B + s(A_{12})) + s(A_{11}), \nonumber
\\
	Z_1' = 4(s(A_{11} \setminus \{ B \} \cup \{ a \}) + B + s(A_{12})) + s(A_{11} \setminus \{ B \} \cup \{ a \}). \nonumber
\end{align}
After the exchange the regret on machine 1 has decreased by $5(B-a)$. Similarly, let $Z_2$ be the regret on machine 2 before the exchange, and $Z_2'$ be the regret on machine 2 after the exchange:
\begin{align}
	Z_2 = 4(s(A_{21}) + a' + s(A_{22})) + s(A_{21}), \nonumber
\\
	Z_2' = 4(s(A_{21} \setminus \{ a \} \cup \{ a' \}) + B + s(A_{22})) + s(A_{21} \setminus \{ a \} \cup \{ a' \}), \nonumber
\end{align}
where $a'$ corresponds to the widest interval $[B-a', B+a']$ of a job on machine 2 before the exchange, $a' < B$. After the exchange the regret on machine 2 has increased by $4(B-a) - a + a'$. Consequently, the new schedule has less total regret by $B - a'$, thus the initial schedule could not be optimal. By repeating the above job exchange operation we conclude that an optimal schedule must have exactly one $[0,2B]$ job on each machine.

Now assume that there is no solution of 4-\textsc{pp}. Then it is not possible to have two 4-job sets with an equal sum on each of $m$ machines, and consequently, at least one machine gives the regret strictly greater than $4B + 9B_j$. Let this be the machine with index $j=1$. To that machine there are assigned two 4-job sets, with the corresponding values $a_i$ denoted $A_{11}$ and $A_{12}$ respectively, so that $s(A_{11}) > s(A_{12})$. Let us denote the total value of these 8 jobs $s(A_{11})+s(A_{12}) = 2B_1$, for some $B_1 > 0$. The regret generated by this machine is:
$$
	Z_1 = 4(s(A_{11}) + s(A_{12}) + B) + s(A_{11}) = 4B + 8B_1 + s(A_{11}).
$$
Since $s(A_{11}) > B_1$, and $\sum_{j=1}^m B_j = C/2$, the total regret is:
\begin{align}
	Z_1 + \sum_{j=2}^m (4B + 9B_j) = 4mB + 9 \sum_{j=2}^m B_j + 8 B_1 + s(A_{11}) \nonumber
\\
	> 4mB + 9 \sum_{j=1}^m B_j = 4mB + 9C/2 = Z^*. \nonumber
\end{align}
It follows that given a polynomial time algorithm for \textsc{interval minmax regret} $P||\sum C_i$ we would be able to decide 4-\textsc{pp} in polynomial time.
\end{proof}

\section{Conclusions}

In this paper, we proved the strong NP-hardness of the \textsc{interval minmax regret} $P||\sum C_i$ problem, the minmax regret version of one of the basic multiprocessor scheduling problems. It was shown how to compute the maximum regret of a schedule in polynomial time, and how the problem on parallel identical machines relates to its single machine variant. An interesting open problem is to settle the complexity status of the single machine version of the considered problem, when the input data is encoded in unary. Another future research direction is to design approximation algorithms for robust scheduling problems that guarantee approximation ratio below two \cite{kasperski20082}, or prove that it is impossible (unless P=NP). 



\end{document}